\documentclass[12pt]{iopart}
\usepackage[utf8]{inputenc}
\usepackage[sc,osf]{mathpazo}
\usepackage[normalem]{ulem}
\usepackage[T1]{fontenc}
\expandafter\let\csname equation*\endcsname\relax
\expandafter\let\csname endequation*\endcsname\relax
\usepackage{amsmath}
\usepackage{amsthm}
\usepackage{graphics}
\usepackage{graphicx}
\usepackage{float}
\usepackage{braket}
\usepackage[colorlinks=true,linkcolor=blue,
urlcolor=blue,citecolor=blue]{hyperref}
\newtheorem{theorem}{Theorem}
\newtheorem{corollary}{Corollary}

\begin{document}
\title{
Necessary condition for information transfer under
simulated parity-time-symmetric evolution
}

\author{Leela Ganesh Chandra Lakkaraju$^{1}$, Shiladitya Mal$^{1, 2, 3, 4}$, Aditi Sen(De)$^{1*}$}
	
	\address{$^1$ Harish-Chandra Research Institute, A CI of Homi Bhabha National Institute,  Chhatnag Road, Jhunsi, Prayagraj - 211019, India\\
 $^2$ Physics Division, National Center for Theoretical Sciences, Taipei 10617, Taiwan\\
	$^3$ Department of Physics and Center for Quantum Frontiers of Research and Technology (QFort),
National Cheng Kung University, Tainan 701, Taiwan\\
$4$ Centre for Quantum Science and Technology, Chennai Institute of Technology, Chennai 600069, India

\ead{$^*$aditi@hri.res.in}
	}
	
\begin{abstract}
Parity-time $(\mathcal{PT})$ symmetric quantum theory 
can broaden the scope of quantum dynamics beyond unitary evolution which may lead to numerous counter-intuitive phenomena, including single-shot discrimination of non-orthogonal states, faster evolution of state than the standard quantum speed limit,  and violation of no-signaling principle.
On the other hand, $\mathcal{PT}$-symmetric evolution can be realized as reduced dynamics of a subsystem in real experiments within the scope of standard QT. In this experimental setup, if one side of a composite system is evolved according to a $\mathcal{PT}$-symmetric way, a non-trivial information transfer can happen, i.e., the operation performed at one side can be gathered by the other side. By considering an arbitrary shared state between two parties situated in two distant locations and arbitrary measurements, we show that the $\mathcal{PT}$-symmetric evolution of the reduced subsystem at one side is not sufficient for this information transfer to occur. Specifically, 
we prove that the information transfer can only happen when the  density matrix and the corresponding measurements contain complex numbers. Moreover, we connect the entanglement content of the shared state with the efficacy of information transfer. We find evidence that the task becomes more efficient with the increase of dimension.

\end{abstract}

\maketitle

\section{Introduction}
Among the postulates of quantum mechanics, the requirement of the Hermiticity property for the observables has the least support from the perspective of physical considerations. It was pointed out that a condition dictated by the fundamental discrete symmetry of the world, i.e., space-time reflection symmetry may lead to a kind of quantum theory \cite{bender'98, bender'02} which is commonly known as parity-time $(\mathcal{PT})$-symmetric quantum theory. A Hamiltonian in this theory has real eigenvalues in the symmetry unbroken phase and pairwise complex eigenvalues in the symmetry broken phase while in between at the exceptional point, a new type of critical behavior emerges \cite{bender'99, dorey'01, dorey'07, bender'68, bender'69, heiss'12, bender'17, ashida'17}. Note that the Hamiltonian commutes with the $\mathcal{PT}$ operator (product of parity and time operator) in both phases. However, like the Hermitian quantum mechanics, in the unbroken phases, Hamiltonian and the $\mathcal{PT}$ operator possess a common set of eigenvectors which is not the case in the broken phase, thereby breaking the $\mathcal{PT}$-symmetry for the eigenvectors. 

This complex extension of quantum theory not only solves a long-standing problem of negative norm 'ghost' state of the  renormalized Lee model in quantum field theory \cite{bender'07} but also supposed to enlarge the scope of allowed dynamics, thereby culminating in the exploration of varieties of rich phenomena. Throughout the past decade, $\mathcal{PT}$-symmetric Hamiltonian has been realized in  classical systems like electronic circuit \cite{nbender'13}, waveguide \cite{aguo'09, doppler'16}, microcavity \cite{peng'14, chang'14} and in photonics, capitalizing on the fundamental structure of balanced loss and gain, it creates almost a new paradigm of  devices, eg., synthetic photonic lattices \cite{synthetic'12}, single-mode laser \cite{feng'14, miao'16}, high-sensitivity sensors \cite{wchen'17}, wireless power transfer \cite{wlespower'17} to name a few. 

In the quantum $\mathcal{PT}$-symmetric domain,  behaviors qualitatively distinct from the classical domain have been discovered, which include wormhole-like behavior in quantum brachistochrone problem \cite{bbrody'07}, single-shot discrimination of non-orthogonal states \cite{bbrody'13}. Several no-go theorems, valid in standard quantum mechanics, are shown to be disturbed in $\mathcal{PT}$-symmetric quantum theory \cite{ynchen'14, pati'14, acta'16}. More prominently, the weaker condition, no-signaling, which is desired to be satisfied in any physical theory, has been shown to be violated in the Gedanken experimental set-up \cite{lee14}. Later nature of the apparent violation of fundamental no-go theorems has been discussed in detail \cite{nori'19}.

Moreover, it was shown that
if one modifies the inner product, known as the charge-parity-time $(\mathcal{CPT})$ inner product, the $\mathcal{PT}$ symmetric dynamics can restore the standard quantum mechanics \cite{brody'16, japar'17, nori'19, footnote} (cf. \cite{mostafa'01, mostafa'07, sarah'15}). It is noticed that an experimentalist cannot distinguish between $\mathcal{PT}$-symmetric closed evolution and $\mathcal{CPT}$-modified theory, as the metric operator is not observable and there are physically inaccessible degrees of freedom that do not let local observers alter their inner products \cite{brody'16, footnote}.
 Several  interesting recent works, which proved to gain an advantage under the $\mathcal{PT}$-symmetric Hamiltonian by considering the original Dirac inner product,  include the increase of bipartite entanglement \cite{chen'14},  \(\mathcal{PT}\)-symmetric Liouvillian dynamics \cite{caspel'18, dangel'18}, the maximal violation of Leggett-Garg inequalities \cite{Usha'21}.

The selection of the Dirac inner product in our study 
also facilitates the calculation of experimentally measurable probabilities. Notably, the adoption of this formalism has been a common choice in various investigations concerning probability calculations \cite{lin2016non, PhysRevX.4.041001} and the determination of observables in non-Hermitian $\mathcal{PT}$-symmetric systems.

On the other hand, the experimental realization of an effective $\mathcal{PT}$-symmetric evolution of a subsystem adopts the open-system framework.
 Employing Naimark's dilation \cite{gunther'08}, this simulation of $\mathcal{PT}$-symmetric evolution of a local system can be embedded in a subspace of a higher-dimensional Hermitian Hamiltonian \cite{gunther'08}, thus enabling experimental simulation \cite{guang17, guang'20}. See also  Ref. \cite{LeeWu'19}, in which $\mathcal{PT}$-symmetric evolution has been implemented through weak measurement formalism. This is within the scope of standard quantum formalism, which necessitates the use of Dirac inner product. The feasibility of simulating $\mathcal{PT}$-symmetric evolution as a reduced dynamics of a subsystem opens up a better insight and new potential applications as well. In this case, one cannot have advantage beyond standard quantum theory, contradicting with fundamental no-go theorems, but can have other advantages. For example,  it was  shown that for information retrieval task, simulated $\mathcal{PT}$-symmetric evolution provides better control over other existing alternative methods like the quantum Zeno effect or dynamical decoupling method \cite{Kawabata'17}.
 
 In this open system approach, there is no violation of no-signaling principle. 
In particular, it has been shown \cite{guang17} that there is a certain information transfer occurs between Alice and Bob under $\mathcal{PT}$-symmetric non-Hermitian evolution. Specifically, when Alice performs an operation by choosing it randomly on her part of the shared composite system, Bob could infer her choice of input. This may appear to be a violation of no-signalling, but it is not, since the information can only be inferred when Bob gets to know that Alice's subsystem has successfully undergone $\mathcal{PT}$-symmetric evolution. We will work in this open quantum system paradigm, which is within the scope of standard quantum theory, and, therefore,  the Dirac inner product can be a natural choice.
 
We now ask the following question -- {\it is $\mathcal{PT}$-symmetric evolution a sufficient condition for the information transfer between Alice and Bob?}
 
 In this work, we answer this question by providing  general conditions for information transfer by considering an arbitrary  two-qubit state between Alice and Bob;  Alice initially performs an arbitrary operation and then followed by a \(\mathcal{PT}\)-symmetric evolution;  Bob then performs measurement on the qubit locally and gains information about the arbitrary operation of Alice. We find that only if the shared state or the measurement operators contain imaginary numbers,  the $\mathcal{PT}$-symmetric evolution leads to the information transfer after post-selection.  It also shows that such a restriction on the shared state implies  that the state might or might not be entangled. 
 %In our work,
 Note that in the \(\mathcal{PT}\)-symmetric framework, we use the Dirac inner product to obtain the results. We further show that increasing dimension can lead to some advantages in the distinguishability process of quantum states \cite{Dadv1,Dadv2,  Dadv3} and like a two-qubit scenario, similar conditions for two-qutrits can be found in terms of correlators for  information transfer. 
 
In addition, following the experimental setup shown in Ref. \cite{guang17}, we provide a proof-of-principle design of a genuine randomness generation protocol. We show that a  pseudo-random bit string consumed at Alice's side, from a source provided by a supplier who may be an eavesdropper, can be transformed into a genuinely random bit at Bob's side. The genuine string would be half the length of the string at Alice's side. This is possible in an ideal scenario only when $\mathcal{PT}$-symmetric evolution occurs on the underlying complex Hilbert space.

The paper is organized as follows. In Sec. \ref{sec_set}, we describe the prerequisites that are required to present the main results. We discuss in detail the necessity of imaginary numbers in quantum states and measurements on information transfer under $\mathcal{PT}$-symmetric evolution in Sec. \ref{sec_effectasol}. The consequence of Bob's reduced state due to local $\mathcal{PT}$-symmetric operation on Alice's end is illustrated in \ref{sec_effectreduce}. In Sec. \ref{sec_high}, we discuss the enhancement in information transfer when a higher dimensional $\mathcal{PT}$-symmetric system is considered. Before concluding in Sec. \ref{sec_conclu}, the application of the $\mathcal{PT}$-symmetric dynamics towards quantum randomness generation is presented in Sec. \ref{sec_random}.

\section{ $\mathcal{PT}$-symmetric Dynamics and its Quantum Simulation}
\label{sec_set}

Let us begin with a short discussion of a few preliminary notions relevant to the subsequent sections.

\emph{$\mathcal{PT}$-symmetric dynamics.} In the $\mathcal{PT}$-symmetric evolution,  the Hamiltonian satisfies $H=H^{\mathcal{PT}}$, i.e., the Hamiltonian commutes with parity and time reversal operators. In the suitable representation pertaining to our analysis, the parity operator is given by $\mathcal{P}=\sigma_x$. The
time reversal operator, $\mathcal{T}$, actually represents complex conjugation part. 
For the two dimensional Hamiltonian
considered in this work, it follows that $\mathcal{T}\sigma_y\mathcal{T}^{-1} = -\sigma_y$, whereas $\sigma_x$ and $\sigma_z$ remain invariant under the action of \(\mathcal{T}\).

Mathematically, it can be seen that $[H,\mathcal{PT}]=0$ such that $(\mathcal{PT})H(\mathcal{PT})^{-1} = H$. It is more convenient to check that $\mathcal{P} H \mathcal{P} = H^\dagger$, with \(H^{\dagger}\) being the Hermitian conjugation of \(H\).
The $2 \times 2$ Hamiltonian, displaying $\mathcal{PT}$-symmetry,  is described \cite{bender_faster} as 
\begin{equation}
H_{general}^\mathcal{PT}(s,r,\alpha)=\left(\begin{array}{cc}
r e^{i \alpha} & s \\
s & r e^{-i \alpha}
\end{array}\right), \quad r, \, s, \, \alpha \in \mathrm{R},
\end{equation}
for which the eigenvalues are, $E_{ \pm}=r \cos (\alpha) \pm \sqrt{s^2-r^2 \sin ^2(\alpha)}$. This shows that when $s^2 > r^2 \sin ^2(\alpha)$, the Hamiltonian is in the unbroken phase with real eigenvalues, and when  $s^2 < r^2 \sin ^2(\alpha)$, the Hamiltonian is in the broken phase with complex eigenvalues.

The above Hamiltonian can be modified further as \cite{gunther_dilation} 
\begin{equation}
H_{general}^\mathcal{PT}(s,s,\alpha)=\cos(\alpha)I_2+s\left(\begin{array}{cc}
i \sin (\alpha) & 1 \\
1 & -i \sin (\alpha),
\end{array}\right)
\end{equation}
with \(r=s\). 
The choice of parameters  sets the imaginary part of the eigenvalues to be zero, also known as exact $\mathcal{PT}$-symmetry.  
Ignoring the constant scaling in the first term, we consider the  two-dimensional non-Hermitian $\mathcal{PT}$-symmetric Hamiltonian  \cite{lee14} as
\begin{equation}
H=s\left(\begin{array}{cc}
i \sin \alpha & 1 \\
1 & -i \sin \alpha
\end{array}\right),
\label{hamiltonian}
\end{equation}
where $s$ is the scale factor and $\alpha$ is the  non-Hermiticity parameter, such that when $\alpha = n \pi$, $H$ reduces to a Hermitian Pauli matrix ($\sigma_x$). The parameter $\alpha$ ranges in $ [\frac{-\pi}{2},\frac{\pi}{2}]$ such that $|\alpha| = \frac{\pi}{2}$ is the exceptional point. At the exceptional point,  the eigenvectors coalesce together. It is to be noted that the parity operator ($\mathcal{P}$) has $\pm 1$ eigenvalues since $\mathcal{P}^2 = I$. The three Pauli operators have the same property although $\sigma_x$ turns out to display the relevant symmetry with respect to the Hamiltonian considered here.   

Note that the above two-dimensional Hamiltonian has been used in different contexts ranging from the effective Hamiltonian in the open quantum systems to the simulation of \(\mathcal{PT}\)-symmetric systems with photonic devices \cite{ptsame1,ptsame2,ptsame3,ptsame4}.  Let us discuss how this Hamiltonian can represent an effective Hamiltonian describing the evolution of a system that interacts with an external environment. In the context of the system being described, let us consider the system Hamiltonian denoted by \(H^S\), which is connected to an external environment characterized by a decay rate, \(\Gamma\). The dynamics of the system can be described by the Lindblad master equation \cite{petrubook} as
\[
\frac{d\rho}{dt} = -i[H^S,\rho] + \frac{\Gamma}{2}[L\rho L^\dagger - \frac{1}{2}(L^\dagger L\rho + \rho L^\dagger L)],
\]
where \(\rho\) represents the density operator of the system and \(L\) denotes the Lindblad operators. This equation accounts for the unitary evolution governed by the system Hamiltonian as well as the dissipative effects caused by the coupling to the environment. To simplify this expression, we introduce the effective Hamiltonian, \(H_{\text{eff}}\),  as
\[
H_{\text{eff}} = H^S - \frac{i\Gamma}{2}L^\dagger L.
\]
The effective Hamiltonian captures the non-Hermitian nature of the system due to the dissipative processes. It incorporates the original system Hamiltonian while accounting for the decay induced by the environment. Additionally, the term \(\frac{\Gamma}{2}L\rho L^\dagger\) in the Lindblad master equation is referred to as the jump operation. It represents the contribution to the system dynamics caused by quantum jumps, which occur due to the interaction with the environment. 

However, in the semiclassical limit, it is often permissible to neglect the effects of the jump operation \cite{fazio_2018}. Consequently, the evolution of the system is governed by the non-Hermitian effective Hamiltonian \(H_{\text{eff}}\). By considering the above simplifications, we can focus on the non-Hermitian Hamiltonian to study the behavior and evolution of the system within the semiclassical approximation.

We now illustrate that the $\mathcal{PT}$-symmetric Hamiltonian in Eq. (\ref{hamiltonian}) can be obtained as the effective description of a system under the action of an amplitude damping channel. Here the system Hamiltonian of a single qubit is given by $H^S = \sigma_x^S $. Physically, it can be realized as a single spin interacting with an electromagnetic field such that the spontaneous emission takes place due to the excited nature of the spin.  The corresponding Lindblad master equation can be read as \cite{open_book} $\frac{d\rho}{dt}=-i[\sigma_x^S ,\rho]+\gamma (\sigma^-\rho\sigma^+-\frac{1}{2}\{\sigma^+\sigma^-,\rho\})$ and the corresponding effective Hamiltonian is given by

\begin{align}
     \nonumber H_{eff} &=  H^S-\frac{i\Gamma}{2}L^\dagger L \\\nonumber& =\sigma_x-\frac{i\Gamma}{2}\sigma^+\sigma^-\\&=\sigma_x-\frac{i\Gamma}{4}\sigma_z,
 \end{align}
where $\sigma^\pm = \sigma_x \pm i \sigma_y$. We drop the $S$ index for brevity. In our calculations, instead of using the dissipation term \( \frac{\Gamma}{4}\), we choose the non-Hermiticity parameter as  \( \frac{\Gamma}{4} = \sin\alpha\). It can be understood that the effective Hamiltonian becomes more non-Hermitian as the strength of the system-bath interaction increases.(see Ref. \cite{fazio_2018}) for similar procedure).

Such a description also provides novel insights into the evolution of the system and can be responsible to obtain some advantages out of quantum  devices. Two such examples include enhanced sensitivity \cite{Hodaei2017} at the exceptional point and thermodynamic advantage in thermal machines \cite{PRXQuantum.2.040346, lakkaraju_konar_de_nhqb}. Moreover, a non-trivial topological effect, known as the non-Hermitian skin effect,  where all the eigenvalues are localized in the boundary  emerges in non-Hermitian systems, thereby violating bulk-boundary correspondence seen in Hermitian counterparts \cite{Tanmoy_das_review_skin_effect_2019}.

\noindent The initial state of the system, \(\rho(0)\), is evolved to a final state,
\begin{equation}
\rho(t)=\frac{\exp (-i H t) \rho(0) \exp(i H t)}{Tr [\exp (-i H t) \rho(0) \exp (i H t)]}.
\end{equation}

Let us now describe  the experimental set-up proposed in  Ref. \cite{guang17} for simulating the effective $\mathcal{PT}$-symmetric Hamiltonian in the open-system framework. A pair of space-like separated photons are generated by the parametric down-conversion which are then sent to two different protagonists say, Alice and Bob via two different channels. After Alice encodes a single bit of information on her photon by randomly choosing between two operations ($A_\pm$), it is evolved according to the $\mathcal{PT}$-symmetric Hamiltonian.  It is implemented by evolving the photon together with an auxiliary system via conventional quantum gate operation and subsequently, performing a measurement on the auxiliary system. Depending upon a particular post selection of the auxiliary system, Alice's photon evolved according to $\mathcal{PT}$-symmetric Hamiltonian. It is to be noted here that the entire protocol follows the  standard quantum theory and no genuine $\mathcal{PT}$-symmetric Hamiltonian is realized although it can be  simulated. On the other hand, Bob's photon is evolved according to the identity channel. Finally, Alice and Bob measure  locally on their subsystems simultaneously and record some outcomes.

\emph{Let us now describe the no-signaling condition and the appropriate figure of merit for the information transfer conditioned upon post-selection.} As a consequence of relativistic causality, information cannot be transferred faster than light between two space-like separated observers, i.e., each space-like separated  party cannot predict the operation choice of the other from his/her own measurement statistics. This restriction on allowed joint probabilities is known as no-signaling condition in the literature \cite{RevModPhys.86.419}. 
Let us suppose now that two observers, Alice and Bob, share a maximally entangled state, $|\psi^+\rangle = \frac{|00\rangle+|11\rangle}{\sqrt{2}}$  and 
 one of the qubits of a shared maximally entangled state, at Alice's end,   is evolved under $\mathcal{PT}$-symmetric effective Hamiltonian after a random operation ($A_{\pm}$), i.e.,  either $\mathbb{I}$ or $\sigma_x$ has been performed at Alice's part   \cite{lee14}. After these operations, initially shared maximally entangled state, $|\psi^+\rangle$, transforms to $\psi_f^{\pm}$ as
$\left|\psi_{f}^{\pm}\right\rangle =\left[(U(t) \otimes I) . (A_{\pm} \otimes I)\right]|\psi^+\rangle.$

Finally, Alice and Bob perform measurement in $|\pm_y\rangle\langle\pm_y|$ basis, (i.e., eigenvectors of $\sigma_y$) on their respective subsystems to obtain joint probability,
\begin{align}
P(a,b|A_{\pm},B)=\langle\psi_f^{\pm}|(|a\rangle\langle a|\otimes |b\rangle\langle b|)|\psi_f^{\pm}\rangle,
\end{align}
where $a$ and $b$ are measurement outcomes of Alice and Bob respectively (for details see  \ref{non-maximal_calc}).

Let us now define the appropriate figure of merit to quantify information transfer as
\begin{align}
\label{eq_asol}
P_+ - P_-=\sum_a P\left(a, b \mid A_{+}, B\right)-\sum_a P\left(a, b \mid A_{-},B\right).
\end{align} 
If the difference vanishes,  the no-signalling condition is preserved and otherwise, information about Alice's operation choice can be inferred by Bob by looking at the measurement statistics. A similar condition also holds for interchanging Alice and Bob. In a similar spirit, the satisfaction of no-signalling principle implies that Bob's reduced system cannot be affected by space-like separated Alice's operation and vice-versa \cite{ujjwalashmita}.
We can use the same condition as a figure of merit for information transfer and check if Bob can infer information about Alice's arbitrary operation setting (see \ref{sec_effectreduce}). In previous works, a maximally entangled state has been used to illustrate the result. We use an arbitrary state and arbitrary measurement to establish a condition for information transmission.

\section{Information transfer under $\mathcal{PT}$-symmetric Hamiltonian requires complex numbers}
\label{sec_effectasol}
In this section,  we show that information transfer under $\mathcal{PT}$-symmetric Hamiltonian never happens if the states and measurements contain only real numbers. 
 
To illustrate it, going beyond maximally entangled shared states, let us consider (i) non-maximally entangled pure states, 
(ii) the Werner state \cite{Werner}, given by $\rho^W = p |\psi^+\rangle\langle\psi^+| + \frac{1-p}{4}I_4$,  with \(|\psi^+\rangle\) and \(I\) being the maximally entangled state and white noise respectively, and
(iii) arbitrary two-qubit density matrix. 

Notice that the set of states in (i) and (ii) is the subset of (iii). However, before considering the general framework, we want to observe whether the protocol of information transmission with maximally entangled states can be translated to other classes of states or not. 
In this respect, natural choices are a set of states considered in (i) and (ii). In particular, entanglement increases monotonically with the varying parameter of non-maximally entangled pure states while the Werner state shows interesting entanglement features which cannot be detected via Bell inequalities \cite{Werner}. 

Furthermore, in an experiment,  the states typically produced are noisy and are close to a maximally entangled state which can be well approximated by the Werner state. 

Interestingly, such a general framework provides us a necessary condition for information transfer i.e., we find that  if states and observables both contain only the real numbers, information gain of Bob about Alice's random operations cannot  happen. Specifically, it means that  the operators and states should contain complex numbers, in order for information transmission to occur in the presence of $\mathcal{PT}$-symmetric evolution.
 
Before presenting the main results for arbitrary two-qubit states with arbitrary measurement, let us illustrate our findings with shared non-maximally entangled,  Werner, and two-qubit states by performing the local projective measurement in the \(y\)-direction, i.e., in the basis \(|\pm_y\rangle = \frac{|0\rangle \pm i |1\rangle}{\sqrt{2}}\)  as described in Ref. \cite{lee14}.  

For the sake of completeness, we reiterate the choice of inner product.  In this work, the modified inner product $\mathcal{CPT}$ \cite{brody'16, japar'17, nori'19, mostafa'01, mostafa'07, sarah'15} is not used when calculating probabilities and expectation values; instead, the conventional Dirac inner product is used. This is because we conduct our study under an open-system paradigm of simulating the $\mathcal{PT}$-symmetric Hamiltonian. In this setting, the evolution of the entire system (i.e., system along with environment) is governed by unitary dynamics and the reduced system part evolves under a $\mathcal{PT}$-symmetric Hamiltonian. As this scenario is well within the standard quantum formalism, the choice of inner product is obvious. On the other hand, the $\mathcal{CPT}$ inner product must be considered if one works within a closed system framework having $\mathcal{PT}$-symmetry.

\emph{Non-maximally entangled  and Werner  states under $\mathcal{PT}$-symmetric local evolution. }
Suppose Alice and Bob initially share the  state, 
\(|\psi^{+}\rangle = \frac{\beta|++\rangle +\gamma|--\rangle}{\sqrt{\beta^2 + \gamma^2}}\),

where $|\pm\rangle$ are eigenstates of $\sigma_x$.
 Depending upon her random information, Alice applies either $A_{+}=I$ or $A_{-}=\sigma_{x}$ which is followed by the non-unitary evolution, $U(t)$, generated by the non-Hermitian Hamiltonian in Eq. (\ref{hamiltonian}),
 \begin{equation}
U(t) \equiv e^{-i t H}=\frac{1}{\cos \alpha}\left(\begin{array}{cc}
\cos \left(t^{\prime}-\alpha\right) & -i \sin t^{\prime} \\
-i \sin t^{\prime} & \cos \left(t^{\prime}+\alpha\right)
\end{array}\right).
\end{equation}
 Here $t^{\prime} = \frac{\Delta E}{2} t$, where \(\Delta E = E_{+} - E_{-}\) and we study the evolution for a specific time $t = \frac{\pi}{\Delta E}$, which leads to the evolution operator, given by

\begin{equation}
U(t) = \left(\begin{array}{cc}
\sin \alpha & -i \\
-i & -\sin \alpha
\end{array}\right).
\label{timeEvol}
\end{equation}
After the $\mathcal{PT}$-symmetric evolution, the normalized joint state becomes

\begin{equation}
\left|\psi_{f}^{\pm}\right\rangle =\left[(U(t) \otimes I) . (A_{\pm} \otimes I)\right]|\psi^+\rangle.
\end{equation}
It is to be noted again that here we follow the conventional inner product as in open system quantum simulation, everything follows according to standard quantum theory.
After Alice and Bob measure in the $|\pm_y\rangle$ basis, we obtain 
\begin{equation}
P_+ - P_- = \frac{8 \beta \gamma \sin \alpha}{\left(\beta^{2}+\gamma^{2}\right)(-3+\cos 2 \alpha)}.
\end{equation}
Here we observe that the difference vanishes when either $\alpha = n\pi$,  or when \(\beta =0\) or  \(\gamma = 0\). It implies that  when the state is a product state, even with \(\alpha\ne n \pi\), the non-Hermitian Hamiltonian does not lead to information transfer or as we know, when the evolution is unitary, Bob cannot get the information about Alice's operations, even if the shared state is entangled. 

In the case of Werner states, \(\rho^W\) \cite{Werner},
following the same prescription as before,  the information about Alice's subsystem  can be protected from Bob when
\begin{equation}
P_+-P_- = \frac{4 p \sin \alpha}{-3+\cos 2 \alpha} 
\end{equation}
vanishes, i.e. even when $\alpha \neq n\pi$ but $p = 0$, implying  the shared state is a maximally mixed state. 
Both the examples illustrate that maximally entangled states are not special and hence the condition of information transmission for shared arbitrary two-qubit states can be important to obtain which we will address now. 

\subsubsection{Condition for arbitrary two-qubit states} 

The canonical form of an arbitrary two-qubit state  can be written as
\begin{align}
\rho (m_i, m'_i, C_{ii}) &=  
\frac{1}{4}(I + \sum_{i=x,y,z} [ m_i\sigma_i \otimes I \nonumber\\
&+m_i^\prime I \otimes \sigma_i
+C_{ii} \sigma_i \otimes \sigma_i]),
\label{ArbiMix}
\end{align}
where \(m_i  = \mbox{tr} ((\sigma_i \otimes I) \rho)\), \( m^{\prime}_i = \mbox{tr} ((I \otimes \sigma_i )\rho)\) and \(C_{ii} = \mbox{tr} ((\sigma_i \otimes \sigma_i) \rho)\) denote magnetizations and correlators  respectively.
Following the similar procedure, we obtain that  when $C_{yy} = m_y m'_y$ or when \(C_{yy} = m_y =0\), the difference,
\begin{align}
\label{expression}
%\nonumber\\
 P_+ - P_- &= \nonumber \\ 
 &\frac{2(C_{yy}-m_y  m_y^\prime )(-3+\cos 2 \alpha) \sin \alpha}{(-3+\cos 2 \alpha+4 m_y \sin \alpha)\left(1+2 m_y \sin \alpha+\sin^{2} \alpha\right)}
 \end{align}
vanishes even when the $\mathcal{PT}$-symmetric evolution occurs at Alice's port.
 It indicates  that the probabilities exclusively depend on  state variables in the $y$ direction  such as $m_y, m_y^\prime \text{ and } C_{yy}$. This is because the measurement is performed in the direction of $|\pm_y\rangle\langle\pm_y|$. Only the $\sigma_y$ component is contracting with the measurement setting, hence the corresponding variables are visible in the probabilities. Notice that at $|\alpha|=\frac{\pi}{2}$, the expression in the first parenthesis of the denominator  in Eq. (\ref{expression}) vanishes for $m_y = 1$ which exemplifies the effect of exceptional point and non-diagonalizability of the Hamiltonian.
 
The entire analysis is performed by considering the measurement of Bob in the $y$-direction. It may suggest that our result heavily depends on Bob's measurement choice. Thus such constraint on measurement is removed by performing the measurement in an arbitrary direction and we will now prove that information transfer requires imaginary terms in states as well as measurements.

\begin{theorem}
When Alice and her distant partner, Bob share a quantum state in which 
(a) Alice’s subsystem undergoes evolution governed by $\mathcal{PT}$-symmetric Hamiltonian and (b) if the shared state and the measurement performed by Alice and Bob contain complex numbers,     information about Alice’s operations, chosen randomly by her, can be predicted successfully by Bob with a finite probability. 

\end{theorem}

\begin{proof}
By considering an arbitrary two-qubit state, \(\rho (m_i, m_i^{'}, C_{ii})\), Alice applies \(I\) and \(\sigma_x\) (without loss of generality, we can choose such a fixed operation, since at the end, Alice performs arbitrary measurements), followed by an application of  local $\mathcal{PT}$-symmetry operation at Alice's node. Finally, Alice  locally  performs measurements on the basis
 $\{|\phi\rangle\langle\phi|,  |\phi^\perp\rangle\langle\phi^\perp|\}$ with $|\phi\rangle = \cos \frac{y}{2}|0\rangle+e^{iv} \sin(\frac{y}{2})|1\rangle$ and $|\phi^{\perp}\rangle = \sin \frac{y}{2}|0\rangle - e^{-iv} \cos \frac{y}{2}|1\rangle$ while Bob measures
 $\{|\varphi\rangle \langle \varphi|, |\varphi^\perp \rangle \langle \varphi^\perp|\} $ with $|\varphi\rangle = \cos \frac{z}{2} |0\rangle+e^{iu} \sin \frac{z}{2} |1\rangle$.  
The difference between the probability of obtaining   $|\varphi\rangle$ at Bob's end, reads
\begin{align}
P_+^a - P_-^a &=  \frac{-2 (7 \sin \alpha-\sin 3 \alpha)}{((-3+\cos 2 \alpha)^{2}-16 m_y^{2} \sin^{2} \alpha) } \nonumber \\
& \times [ m_y  m'_x \cos u \sin z 
+\sin u \sin z ( -C_{yy} + m_y m'_y) \nonumber \\
&+ m_y m'_z \cos z]
\label{eq_condikamal}
\end{align}
%\end{widetext}
where superscript, ''a'', is used to indicate arbitrary measurements (see \ref{non-maximal_calc} and Sec. \ref{data_avail} for detailed calculation).

This proves that when the states and measurements contain real numbers, i.e., \(m_y\), $m_y^\prime$ and \(C_{yy}\) vanish or the phase in the measurement at Bob's node vanish, i.e.,  $u=0$, $\mathcal{PT}$-symmetry evolution does not lead to information transmission from Alice to Bob. Notice that when $m_y$, $m_y^\prime$ and $C_{yy}$ are non-zero, the presence of $\sigma_y$ operator containing imaginary numbers ensures the involvement of complex numbers in the state while non-vanishing $u$ in the measurement operators also include complex numbers with a non-zero imaginary component. Here, requirements of complex numbers mean that either states or measurements or both must contain complex entry independent of the choice of basis to express them. Moreover, we observe that the numerator vanishes when $7 \sin \alpha - \sin 3\alpha = 0$, which is when  $\alpha = n\pi$ and also when $\alpha = 2 \pi n \pm 2 i \tanh ^{-1}(\sqrt{3\pm2 \sqrt{2}}), \quad n \in \mathbb{Z}$.
 
  \end{proof}
 Interestingly, one can show that when $C_{yy}, m_y \text{ and } m_y^\prime$ vanish,  the shared state is unentangled since the state \(\rho\) coincides with its partial transposed state with respect to a subsystem, say, Alice \cite{Peres'96, Horodecki'96}.  Since the time evolution operator at the given specific time can be written as $U = \sin \alpha \sigma_z - i\sigma_x$,  we find that in \(P_+^a - P_-^a \),  \(m_x, m_z, C_{xx}\) and \( C_{zz}\) are not present.

\begin{corollary}
For the maximally entangled state and for the Werner states, arbitrary measurements at Bob's end are required to be complex for gathering information about Alice's operation by Bob. 
\end{corollary} 
\begin{proof}
 When $|\varphi\rangle$ at Bob's side clicks, the condition for the non-maximally entangled state is
\begin{equation}
P_+^a - P_-^a = \frac{8 \beta \gamma \sin u \sin z \sin \alpha}{\left(\beta^{2}+\gamma^{2}\right)(-3+\cos 2 \alpha)}
\end{equation}
while for the Werner states, it can be given by
\begin{equation}
P_+^a - P_-^a =\frac{4 p \sin u \sin z\sin \alpha}{-3+\cos 2 \alpha}
\end{equation}
(in Eq. (\ref{eq_condikamal}), putting \(m_y\) =0, and \(C_{yy} = -p\)). Clearly, information transfer cannot happen to Bob if the measurement at Bob's side is real, i.e., $u = 0$ or the measurement is along the \(z\)-direction or  $\alpha = n \pi$. 

For pure states, we also get the condition that when the state is a product state, information transfer cannot occur while for Werner states, the condition gives the state to be maximally mixed states. Interestingly, note that the information gain at Bob's end is not related to entanglement for mixed shared states since the Werner state is entangled with \(p > 1/3\) (see Appendix and Fig. \ref{fig_werner2}).   
\end{proof}
\subsubsection{Distinguishing Bob's states via trace distance after PT-symmetric evolution at Alice's node} 

To exhibit the results further, let us find the distance between Bob's state corresponding to Alice's action of $I$ or $\sigma_x$ on the shared state. It was shown that  for maximally entangled state \(|\psi^+\rangle\), at $\alpha = \frac{\pi}{2}$, Bob's state is $|\pm_y\rangle\langle\pm_y|$ corresponding to either $I$ or $\sigma_x$ operation on Alice's side \cite{lee14}, which indicates that  Bob can distinguish his own subsystem perfectly, thereby distinguishing Alice's operation with unit probability.  Here notice that the probability of success in distinguishing Alice's operation  also depends on  the probability involved in post-selection. In general, Bob's state may not always  be perfectly distinguishable, it can only be discriminated probabilistically, and hence  with a non-vanishing probability, Bob can gain information about Alice's randomly chosen operation. We also emphasize here that calculating trace distance is a quantifiable way of distinguishing Bob's states  when Alice's state has undergone $\mathcal{PT}$-symmetric evolution.

The trace distance between two density matrices, \(\rho\) and \(\sigma\) can quantify the maximum probability of minimum error discrimination between the states by the best quantum measurement strategy \cite{Helstrom}  and  is defined as

\begin{equation}
T(\rho_1, \rho_2)=\frac{1}{2} \operatorname{Tr}\left[\sqrt{(\rho_1-\rho_2)^{\dagger}(\rho_1-\rho_2)}\right]=\frac{1}{2} \sum_{i}\left|\lambda_{i}\right|,
\end{equation}
where \(\lambda_i\)s are the eigenvalues of the operator \(\sqrt{(\rho_1-\rho_2)^{2}}\). In this situation,  $\rho_1$  and \(\rho_2\) represent  the states of Bob when $I$  and \(\sigma_x\) are applied by Alice respectively. For an arbitrary two-qubit density matrix, it reduces to 
%\begin{widetext}
\begin{align}
T = & \left| \frac{\sqrt{C_{yy}^{2}+ m_x^{\prime 2} m_y^{2} - 2 C_{yy} m_y m_y^\prime + m_y^{2} m_y^{\prime 2}+ m_y^{2} m_z^{\prime 2}}}{\left(-1+2  m_y  \sin\alpha-\sin^{2} \alpha\right)\left(1+2 m_y \sin \alpha +\sin^{2} \alpha \right)}   \right| \nonumber \\
& \times | (\sin \alpha+\sin^{3} \alpha)|,
\end{align}
%\end{equation}
%\end{widetext}
which again shows that if both $C_{yy}$ and $m_y$ vanish,  the distance vanishes. It also implies that if the shared state  has no imaginary component, the information about Alice's operations cannot be gathered by Bob, even probabilistically, thereby confirming Theorem 1. 

Moreover, we will show that  this approach of distinguishability to address the information transfer criteria with \(\mathcal{PT}\)-symmetric evolution can be easily generalized to higher dimensions.

\section{ Higher information transfer in Higher dimension}
\label{sec_high}

\begin{figure}[h]
\includegraphics[width=0.95\linewidth]{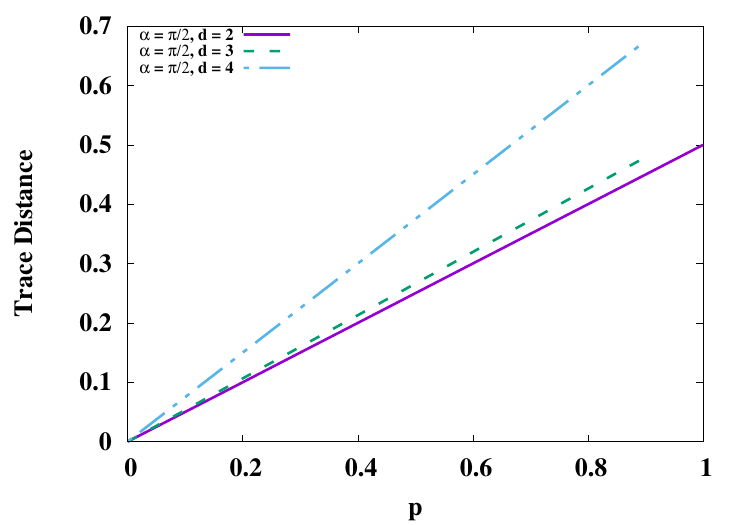}
\caption{Trace distance between the initial states and the state after $\mathcal{PT}$-symmetric evolution (vertical axis)) against $p$  (horizontal axis) of the Werner state, given by \(\rho^W = p |\psi^+\rangle \langle \psi^+  + (1 -p) I\) where \(|\psi^+\rangle\) is the maximally entangled state for a given dimension while \(I\) denotes the identity operator of that dimension.  We choose different \(\alpha\) values for demonstration. Both axis are dimensionless. }
\label{fig_higherdim}
\end{figure}

Up to now, we consider the shared states to be two-qubits. Let us now move to a higher-dimensional situation.   
The $\mathcal{PT}$-symmetric Hamiltonian, with spin-$1$ matrices, takes the form as \cite{highd}

\begin{equation}
\label{Spin1}
H=\frac{1}{\sqrt{2}}
\left(\begin{array}{lll}
i\sin \alpha & 1 & \quad 0 \\
1 & 0 & \quad 1 \\
0 & 1 & -i\sin \alpha\\
\end{array}\right).
\end{equation}
Suppose Alice and Bob share a two-qutrit density matrix, \(\rho^3 = \frac{1}{9}(I + \sum_{i=x, y, z} m_i S_i \otimes  I + m'_i I \otimes S_i + C_{ii} S_i \otimes S_i)\), where \(S_i\) denotes the spin-$1$ matrices, and \(m_i, m'_i, C_{ii}\) are correspondingly local and global correlators defined in terms of the spin operator, \(S_i\). 
If we  now compare Bob's state before and after applying $\mathcal{PT}$-symmetric evolution on Alice's side, we find that  they coincide when  $m_y$, $m'_y$, and $C_{yy}$ are vanishing even with  $\alpha \neq n\pi$ which is similar to the qubit case as shown in Appendix. However, unlike the qubit case, we also require $m_z$ and $C_{zz}$ to be non-zero in order to allow information transfer. 

Let us now  analyze the effects of dimensionality on distinguishability i.e., with the increase of dimension, we compute the trace distance  between states before and after $\mathcal{PT}$-symmetric evolution, for different \(\alpha\) values when the shared state is the Werner-like state in the respective dimension. We observe that the difference increases with the increase of the dimension as depicted in Fig. \ref{fig_higherdim}, when \(\alpha\) is chosen to be close to \(\pi/2\). Therefore,  it displays that at higher dimensions, information transfer can be more compared to the situation in lower dimensions.

\section{Genuine Randomness Generation via $\mathcal{PT}$-symmetric evolution} 
\label{sec_random}

Let us now describe how the  results obtained in the preceding section enable  us to devise a proof of principle  application of $\mathcal{PT}$-symmetric dynamics in  quantum information processing tasks, specifically, in the quantum cryptographic domain. It is known that in the classical regime, only pseudo-random numbers can be obtained while in the quantum domain, genuine randomness can be certified by using quantum no-go theorems \cite{pironio, shiladitya}.   In randomness generation, a shorter genuine random bit-string is generated from a pseudo-long random string \cite{renner} (cf. \cite{experiment}).   Toward  executing the protocol in our set-up,  we replace the short-delay quantum random pulse generator employed in the original protocol \cite{guang17} with a weak random number generator (WRNG) like Santha-Vazirani source \cite{santha, renner}. The choice is due to the fact that we want to show genuine random number generation (RNG), thereby showing a perfect quantum random  number source, and hence  if parties have initially quantum random number generator in their possession, the protocol does not make sense. Here we also assume that this source provided by a supplier  may be  an eavesdropper. 

Notice that although the design of the protocol is not new in the literature, it demonstrates a quantum information protocol based on $\mathcal{PT}$-symmetric dynamics.  The primary motivation for presenting this application is to establish a connection between the length of the generated random string and the non-vanishing nature of classical correlators and magnetization in the $y$-direction, namely \(m_y\), $m_y^\prime$ and \(C_{yy}\), as demonstrated in Theorem 1. This relation is indirectly established through the ability of Bob in distinguishing between the states before and after $\mathcal{PT}$-symmetric evolution on Alice's part. The distinguishability is proportional to the trace distance as illustrated in Fig. \ref{fig_higherdim} which is non-vanishing when the conditions stated in Theorem 1 are satisfied. To establish such a relation in an applicable context, we have chosen a randomness-generating protocol.\\
Let us first give the protocol   for a maximally entangled shared state and then we go for arbitrary shared two-qubit states.  Before starting the protocol, we assume that Alice and Bob's actions are predecided. Establishing shared randomness between two parties, Charu and Bob,  we require a third party, Alice, who has no knowledge about the random data. Suppose  Charu, \(C\) holds some random bit string in advance and wants Bob to have a share of it. Moreover, \(C\) has a limited power, i.e., only possesses single qubit gates. She, therefore, asks Alice to help her in sending the random bits to Bob but do not want Alice to learn the bit values. So Alice has to send random bits possessed by \(C\) to Bob without knowing it.\\
Charu encodes one bit of information obtained from WRNG on  her subsystem of the shared maximally entangled photon via applying $\sigma_x$ when she gets $0$ and does nothing when she has $1$.  After this operation, she sends her qubit to Alice. Using the open system approach, Alice jointly evolves her photon together with the auxiliary system by appropriate unitary  followed by post-selection on auxiliary subsystem (for details, see \cite{guang17}). After a short delay, Bob measures $\sigma_y$ on his part and records his outcome. 
After that, he receives a phone call from Alice and learns about the incident when she succeeds in the post-selection to obtain a particular outcome. Bob can keep only those outcomes and discard all others. At the end of the protocol on average, half of the time Alice  succeeds in simulating $\mathcal{PT}$-symmetric evolution in the open system setup, and Bob has a genuine random bit string which has a length half of the initial string of Alice obtained from Charu. Thus generated randomness is genuine as we do not allow any interference in the laboratories of two trusted parties including the source of entanglement.  

Instead of a maximally entangled state, if Alice and Bob share an arbitrary two-qubit density matrix, \(\rho\), such random number generation is probabilistically successful, and Bob can have a bit string which would have a length smaller than half the original one,  provided the classical correlator and magnetization in the $y$-direction, i.e.,  \(C_{yy}, m_y, m'_y\) are non-vanishing as shown in Theorem 1.  The length of the genuine random bit string would be related to the distinguishability of states as measured in the previous section.

\section{Conclusion}
\label{sec_conclu} 

Many intriguing features have been reported in the field of non-Hermitian quantum physics, particularly in $\mathcal{PT}$-symmetric quantum theory.  In particular, a system evolving under such a $\mathcal{PT}$ symmetric Hamiltonian shows a violation of the no-signaling condition. There are two approaches proposed in the literature to resolve the issue. Firstly, one can modify the inner product completely, i.e., one considers a $\mathcal{CPT}$-inner product. Secondly,  one can resolve this issue by embedding  the $\mathcal{PT}$-symmetric Hamiltonian  into a bigger Hermitian space, which was also experimentally tested. In the latter setup, an information transfer occurs after post-selection, thereby avoiding violation of the no-signalling principle.

In this work, we proved that for arbitrary two-qubit shared states, \(\mathcal{PT}\)-symmetric evolution is not sufficient for the nontrivial information transfer about the operations performed by spatially separated parties. Specifically, we found that for information transmission,  it is necessary that the shared states or the measurements must contain complex numbers. For example, considering the non-maximally entangled pure state and the Werner state, we demonstrated that information can only be obtained when either the state is not product or maximally mixed or the measurement has a complex component.  We indicated that information content about the operations increases with the increase of dimension. We then showed that such a consideration has a direct connection with designing a quantum random number generator from a  pseudo-random bit-string. In particular, we showed that systems under \(\mathcal{PT}\)-symmetric evolution can indeed generate shared randomness which has direct implications in quantum state certification as well as quantum cryptographic protocols.

 \section*{acknowledgement}
We acknowledge the support from the Interdisciplinary Cyber Physical Systems (ICPS) program of the Department of Science and Technology (DST), India, Grant No.: DST/ICPS/QuST/Theme- 1/2019/23 and SM acknowledges Ministry of Science and Technology(Grant no. MOST 111-2124-M-002-013), Taiwan.

\section{Data availability}
\label{data_avail}
The data that support the findings of this study are openly available at the following URL: https://github.com/GaneshChandra/Necessary-condition-for-information-transfer-under-simulated-PT-symmetric-evolution

\onecolumn
\appendix

\section{Information transfer using an arbitrary shared state}
\label{non-maximal_calc}

Let us consider that the shared state between Alice and Bob is an arbitrary two-qubit density matrix, \( \rho (m_i, m'_i, C_{ii})\) as given in Eq. (\ref{ArbiMix}). 
When Alice wants to convey the classical bit, $0$,  she acts $A_+ = I$ on the state, while she performs $A_- = \sigma_x$  in case of $1$.  The corresponding state becomes $\rho_+ = (I \otimes I)\rho(I \otimes I)^\dagger$ and  $\rho_- = (\sigma_x \otimes I)\rho(\sigma_x \otimes I)^\dagger$ for  $A_+$ and $A_- $ respectively.
Now, Alice evolves her qubit with a $\mathcal{PT}$-symmetric Hamiltonian following Eq. (\ref{timeEvol}) as $\rho_\pm(t)=U(t)\rho_{\pm} U^\dagger(t)$ 
where
\begin{equation}
    \rho_+ (t) = \frac{1}{\text{norm}} \left(
\begin{array}{cccc}
 a & b & c & d \\
 e & f & g & h \\
 i & j & k & l \\
 m & n & o & p \\
\end{array}
\right),
\end{equation}
where
\begin{align*}
&a=\frac{1}{8} \left(-\cos (2 \alpha )-\cos (2 \alpha ) m_z^{\text{aa}}+3 m_z^{\text{aa}}-\cos (2 \alpha ) C_{\text{zz}}-C_{\text{zz}}+4 \sin (\alpha ) m_y-\cos (2 \alpha ) m_z-m_z+3\right), \\
&b= \frac{1}{8} \left(-\cos (2 \alpha ) m_x^{\prime}+3 m_x^{\prime}+i \cos (2 \alpha ) m_y^{\prime}-3 i m_y^{\prime}-4 i \sin (\alpha ) C_{\text{yy}}\right), \\
&c= \frac{1}{8} \left(4 i \sin (\alpha )+4 i \sin (\alpha ) m_z^{\prime}+\cos (2 \alpha ) m_x+m_x-i \cos (2 \alpha ) m_y+3 i m_y\right), \\
&d=\frac{1}{8} \left(4 i \sin (\alpha ) m_x^{\prime}+4 \sin (\alpha ) m_y^{\prime}+\cos (2 \alpha ) C_{\text{xx}}+C_{\text{xx}}-\cos (2 \alpha ) C_{\text{yy}}+3 C_{\text{yy}}\right),\\
&e= \frac{1}{8} \left(-\cos (2 \alpha ) m_x^{\prime}+3 m_x^{\prime}-i \cos (2 \alpha ) m_y^{\prime}+3 i m_y^{\prime}+4 i \sin (\alpha ) C_{\text{yy}}\right), \\
&f= \frac{1}{8} \left(-\cos (2 \alpha )+\cos (2 \alpha ) m_z^{\prime}-3 m_z^{\prime}+\cos (2 \alpha ) C_{\text{zz}}+C_{\text{zz}}+4 \sin (\alpha ) m_y-\cos (2 \alpha ) m_z-m_z+3\right), \\
&g = \frac{1}{8} \left(4 i \sin (\alpha ) m_x^{\prime}-4 \sin (\alpha ) m_y^{\prime}+\cos (2 \alpha ) C_{\text{xx}}+C_{\text{xx}}+\cos (2 \alpha ) C_{\text{yy}}-3 C_{\text{yy}}\right),\\
&h= \frac{1}{8} \left(4 i \sin (\alpha )-4 i \sin (\alpha ) m_z^{\prime}+\cos (2 \alpha ) m_x+m_x-i \cos (2 \alpha ) m_y+3 i m_y\right), \\
&i= \frac{1}{8} \left(-4 i \sin (\alpha )-4 i \sin (\alpha ) m_z^{\prime}+\cos (2 \alpha ) m_x+m_x+i \cos (2 \alpha ) m_y-3 i m_y\right), \\
&j= \frac{1}{8} \left(-4 i \sin (\alpha ) m_x^{\prime}-4 \sin (\alpha ) m_y^{\prime}+\cos (2 \alpha ) C_{\text{xx}}+C_{\text{xx}}+\cos (2 \alpha ) C_{\text{yy}}-3 C_{\text{yy}}\right), \\
&k= \frac{1}{8} \left(-\cos (2 \alpha )-\cos (2 \alpha ) m_z^{\prime}+3 m_z^{\prime}+\cos (2 \alpha ) C_{\text{zz}}+C_{\text{zz}}+4 \sin (\alpha ) m_y+\cos (2 \alpha ) m_z+m_z+3\right), \\
&l = \frac{1}{8} \left(-\cos (2 \alpha ) m_x^{\prime}+3 m_x^{\prime}+i \cos (2 \alpha ) m_y^{\prime}-3 i m_y^{\prime}-4 i \sin (\alpha ) C_{\text{yy}}\right),\\
&m= \frac{1}{8} \left(-4 i \sin (\alpha ) m_x^{\prime}+4 \sin (\alpha ) m_y^{\prime}+\cos (2 \alpha ) C_{\text{xx}}+C_{\text{xx}}-\cos (2 \alpha ) C_{\text{yy}}+3 C_{\text{yy}}\right), \\
&n= \frac{1}{8} \left(-4 i \sin (\alpha )+4 i \sin (\alpha ) m_z^{\prime}+\cos (2 \alpha ) m_x+m_x+i \cos (2 \alpha ) m_y-3 i m_y\right), \\
&o = \frac{1}{8} \left(-\cos (2 \alpha ) m_x^{\prime}+3 m_x^{\prime}-i \cos (2 \alpha ) m_y^{\prime}+3 i m_y^{\prime}+4 i \sin (\alpha ) C_{\text{yy}}\right), \\
&p = \frac{1}{8} \left(-\cos (2 \alpha )+\cos (2 \alpha ) m_z^{\prime}-3 m_z^{\prime}-\cos (2 \alpha ) C_{\text{zz}}-C_{\text{zz}}+4 \sin (\alpha ) m_y+\cos (2 \alpha ) m_z+m_z+3\right), \\
&\text{norm} = \sin ^2(\alpha )+2 \sin (\alpha ) m_y+1. 
\end{align*}
 We arrive at the similar expression for 
\begin{equation}
    \rho_{-} (t) = \frac{1}{\text{norm}^{\prime}} \left(
\begin{array}{cccc}
 a' & b' & c' & d' \\
 e' & f' & g' & h' \\
 i' & j' & k' & l' \\
 m' & n' & o' & p' \\
\end{array}
\right).
\end{equation}
where
\begin{align*}
& a^\prime = \frac{1}{4} \left(\sin ^2(\alpha )+\left(\sin ^2(\alpha )+1\right) m_z^{\prime}+\cos ^2(\alpha ) \left(C_{\text{zz}}+m_z\right)-2 \sin (\alpha ) m_y+1\right),\\
 & b^\prime = \frac{1}{4} \left(\left(\sin ^2(\alpha )+1\right) \left(m_x^{\prime}-i m_y^{\prime}\right)+2 i \sin (\alpha ) C_{\text{yy}}\right),\\
  & c^\prime = \frac{1}{4} \left(2 i \sin (\alpha ) \left(m_z^{\prime}+1\right)+\cos ^2(\alpha ) m_x-i \left(\sin ^2(\alpha )+1\right) m_y\right), \\
  &d^\prime = \frac{1}{8} \left(4 \sin (\alpha ) \left(m_y^{\prime}+i m_x^{\prime}\right)+2 \cos ^2(\alpha ) C_{\text{xx}}+(\cos (2 \alpha )-3) C_{\text{yy}}\right), \\
&e^\prime = \frac{1}{4} \left(\left(\sin ^2(\alpha )+1\right) \left(m_x^{\prime}+i m_y^{\prime}\right)-2 i \sin (\alpha ) C_{\text{yy}}\right),\\
 & f^\prime = \frac{1}{4} \left(\sin ^2(\alpha )-\left(\left(\sin ^2(\alpha )+1\right) m_z^{\prime}\right)+\cos ^2(\alpha ) \left(m_z-C_{\text{zz}}\right)-2 \sin (\alpha ) m_y+1\right),\\
  & g^\prime =  \frac{1}{4} \left(2 i \sin (\alpha ) \left(m_x^{\prime}+i m_y^{\prime}\right)+\cos ^2(\alpha ) C_{\text{xx}}+\left(\sin ^2(\alpha )+1\right) C_{\text{yy}}\right),\\
   & h^\prime = \frac{1}{4} \left(\cos ^2(\alpha ) m_x-i \left(2 \sin (\alpha ) \left(m_z^{\prime}-1\right)+\left(\sin ^2(\alpha )+1\right) m_y\right)\right), \\
 &i^\prime =\frac{1}{4} \left(\cos ^2(\alpha ) m_x+i \left(\left(\sin ^2(\alpha )+1\right) m_y-2 \sin (\alpha ) \left(m_z^{\prime}+1\right)\right)\right), \\
  &j^\prime = \frac{1}{4} \left(-2 i \sin (\alpha ) \left(m_x^{\prime}-i m_y^{\prime}\right)+\cos ^2(\alpha ) C_{\text{xx}}+\left(\sin ^2(\alpha )+1\right) C_{\text{yy}}\right),\\
   & k^\prime = \frac{1}{4} \left(\sin ^2(\alpha )+\left(\sin ^2(\alpha )+1\right) m_z^{\prime}-\cos ^2(\alpha ) \left(C_{\text{zz}}+m_z\right)-2 \sin (\alpha ) m_y+1\right),\\
   & l^\prime =  \frac{1}{4} \left(\left(\sin ^2(\alpha )+1\right) \left(m_x^{\prime}-i m_y^{\prime}\right)+2 i \sin (\alpha ) C_{\text{yy}}\right), \\ 
 &m^\prime =\frac{1}{8} \left(4 \sin (\alpha ) \left(m_y^{\prime}-i m_x^{\prime}\right)+2 \cos ^2(\alpha ) C_{\text{xx}}+(\cos (2 \alpha )-3) C_{\text{yy}}\right), \\
 &n^\prime = \frac{1}{4} \left(\cos ^2(\alpha ) m_x+i \left(2 \sin (\alpha ) \left(m_z^{\prime}-1\right)+\left(\sin ^2(\alpha )+1\right) m_y\right)\right), \\ 
 &o^\prime = \frac{1}{4} \left(\left(\sin ^2(\alpha )+1\right) \left(m_x^{\prime}+i m_y^{\prime}\right)-2 i \sin (\alpha ) C_{\text{yy}}\right), \\
 &p^\prime = \frac{1}{4} \left(\sin ^2(\alpha )-\left(\left(\sin ^2(\alpha )+1\right) m_z^{\prime}\right)+\cos ^2(\alpha ) \left(C_{\text{zz}}-m_z\right)-2 \sin (\alpha ) m_y+1\right), \\
 & \text{and} \,\text{norm}^\prime = \sin ^2(\alpha )-2 \sin (\alpha ) m_y+1.
\end{align*}
 Alice  locally  performs measurements in the basis
 $\{|\phi\rangle\langle\phi|,  |\phi^\perp\rangle\langle\phi^\perp|\}$ with $|\phi\rangle = \cos \frac{y}{2}|0\rangle+e^{iv} \sin \frac{y}{2}|1\rangle$ and $|\phi^{\perp}\rangle = \sin \frac{y}{2}|0\rangle - e^{-iv} \cos \frac{y}{2}|1\rangle$ while Bob measures
 $\{|\varphi\rangle \langle \varphi|, |\varphi^\perp \rangle \langle \varphi^\perp|\} $ with $|\varphi\rangle = \cos \frac{z}{2} |0\rangle+e^{iu} \sin \frac{z}{2} |1\rangle$ and similarly \(|\varphi^\perp \rangle\).   

From this, we construct the difference, $P_+ -P_-$, which serves as an appropriate figure of merit for the information transfer between Alice and Bob. Specifically, we obtain
% \begin{equation}
\begin{align}
     P_\pm &=\mbox{Tr}( |\phi\rangle\langle \phi|\otimes|\varphi^{\perp}  \rangle\langle\varphi^{\perp} |\rho_\pm(t)) \nonumber \\
     &+ \mbox{Tr}( |\phi^{\perp}\rangle\langle \phi^{\perp}|\otimes|\varphi^{\perp} \rangle\langle\varphi^{\perp}|\rho_\pm(t)),   
\end{align}
% \end{equation}
where
\begin{align*} 
    &P_+ =\\ 
    & -\frac{\cos 2 \alpha -4 \sin \alpha  (C_{yy} \sin u \sin z + \text{m}_y)+ A \sin z \left(\text{m}_x^\prime \cos u + \text{m}_y^\prime \sin u\right)+ A \text{m}_z^\prime \cos z-3}{4 \left(\sin ^2 \alpha +2 \text{m}_y \sin \alpha +1\right)},
\end{align*}
and
\begin{align*}
    &P_- =\\
    &\frac{\cos 2 \alpha +4 \sin \alpha  (\text{C}_{yy} \sin u \sin z+ m_y +A \sin z \left(m_x^\prime \cos u+m_y^\prime \sin u \right)+\cos 2 \alpha  m_z^\prime \cos z-3 m_z^\prime \cos z-3}{2 (\cos 2 \alpha +4 m_y
 \sin \alpha -3)},
\end{align*}
with \(A =\cos 2 \alpha -3\).

\section{Effects of Alice's operations on Bob's reduced state}
\label{sec_effectreduce}
 
  \begin{figure}[H]
\includegraphics[width=0.95\linewidth]{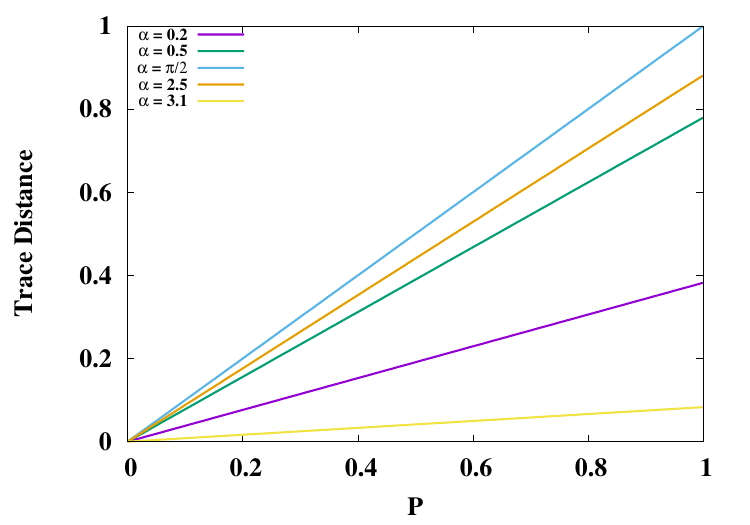}
\caption{ Trace distance  between Bob's state before and after \(\mathcal{PT}\)-symmetric dynamics (ordinate) vs. $p$ (abscissa) of the two-qubit Werner state. Different \(\alpha\) of $\mathcal{PT}$-symmetric   evolution are chosen. Both the axes are dimensionless. }
\label{fig_werner2}
\end{figure}

 There is a straight forward way to check whether the Bob can draw information about Alice's operation, by comparing Bob's state before and after local $\mathcal{PT}$-symmetric Hamiltonian as also shown in Sec. \ref{sec_high}.
 
 \emph{Case 1. Non-maximally entangled states. }
 When the shared state is non-maximally entangled,  Bob's state ($\rho_B$)  is initially
 \begin{equation}
 \rho_B = \left(\begin{array}{cc}
\frac{1}{2} & 1-\frac{2 \beta^{2}}{\beta^{2}+\gamma^{2}} \\
1-\frac{2 \beta^{2}}{\beta^{2}+\gamma^{2}} & \frac{1}{2}
\end{array}\right)
\end{equation}
while 
after the application of   local $\mathcal{PT}$-symmetric Hamiltonian on Alice's subsystem,  Bob's state becomes
%\begin{widetext}
\begin{equation}
\rho_B^\prime = \left(\begin{array}{ccc}
\frac{1}{2} &  U^{nm} \\
\bar{U^{nm}} & \frac{1}{2}
\end{array}\right),
\end{equation}
 where \(U^{nm}=\frac{(\beta^2-\gamma^2)(-3+{\cos}(2 \alpha))+8i \beta \gamma \sin (\alpha)}{4\left(\beta^{2}+\gamma^{2}\right)\left(1+\sin (\alpha)^{2}\right)}\)
%\end{equation}
%\end{widetext}
The condition that  $\rho_B^\prime = \rho_B$ leads to the condition $\gamma = 0$ even when $\alpha \neq n\pi$, which is similar to the condition  reached in Sec. \ref{sec_effectasol}.

\emph{ Case 2. Mixed states.}
For Werner state, Bob's state after $\mathcal{PT}$-symmetry dynamics reads

\begin{equation}
\left(\begin{array}{cc}
\frac{1}{2} & \frac{\text { ì p Sin }(\alpha)}{1+\sin (\alpha)^{2}} \\
\frac{2 \text { i } p\sin(\alpha)}{-3+\cos (2 \alpha)} & \frac{1}{2}
\end{array}\right). 
\end{equation}
from its initially maximally mixed subsystem. 
We plot the trace distance between these two states and find that it vanishes only when $p =0$ and the maximum distance is achieved when \(\alpha = n \pi/2\) (see Fig. \ref{fig_werner2}).  

 Interestingly, we find the similar condition obtained in Eq. (\ref{eq_condikamal}) if we compare Bob's state before the dynamics, 
 \begin{equation}
\rho_B = \left(\begin{array}{rr}
\frac{1}{2}+\frac{m_z^{\prime}}{2} & \frac{m_x^{\prime}}{2}-\frac{i m_y^{\prime}}{2} \\
\frac{m_x^{\prime}}{2}+\frac{i m_y^{\prime}}{2} & \frac{1}{2}-\frac{m_z^{\prime}}{2}
\end{array}\right)
\end{equation}
and after Alice applies local $\mathcal{PT}$-symmetric operation,  
%After just the Local $\mathcal{PT}$ symmetric hamiltonian evolution is applied on Alice, then the Bob's state is,
%\begin{widetext}
 \begin{equation}
\rho_B^\prime = \left(\begin{array}{ccc}
R_+ & U\\
\bar{U} & R_{-}
\end{array}\right)
\end{equation}
%\end{widetext}
with \begin{equation}
R_{\pm} =\frac{1}{2} \left( 1+\frac{\left(1 \pm \sin (\alpha)^{2}\right) {m_z}^{\prime}} {1+2 {m_y} \sin (\alpha)+\sin (\alpha)^{2}}\right),
\end{equation}
and 
\begin{equation}
    U = \frac{\left(1+\sin (\alpha)^{2}\right) {m_x}^{\prime}-{i}\left(2 C_{yy} \sin \alpha+\left(1+\sin (\alpha)^{2}\right) {m_y}^{\prime}\right)}{2\left(1+2 {m_y} \sin \alpha +\sin^{2} \alpha\right)}.
    \end{equation}
Comparing each element of the matrix, we find that 
%A very interesting case emerges here, if the elements of $\rho_B^\prime$ and $\rho_B$ are to be equal, consider the case, of 
%\begin{equation}
\(R_{\pm} \) reduces to \(\frac{1}{2} \pm \frac{m_z^{\prime}}{2} \)
%\end{equation}
when \(m_y\) vanishes while   if  both \(C_{yy}\) and \(m_y\) are vanishing, \(U\) reduces to the off-diagonal term of \(\rho_B\), thereby arriving to the same condition as obtained in Theorem 1.  

\begin{thebibliography}{100}
\bibitem{footnote} The  $\mathcal{CPT}$ inner product does not work for the infinite dimensional systems and it is an open question how to modify the inner product in this domain as noted in \cite{brody'16}
\end{thebibliography}
\end{document}